\documentclass[10pt,conference]{IEEEtran}

\usepackage[english]{babel}
\usepackage{amsmath,amsfonts,amssymb,bbm,amsthm,mdwlist}
\usepackage[ansinew]{inputenc}
\usepackage{latexsym}
\newcommand{\ene}{\mathbb{N}}

\def\cal#1{\fam2#1}

\newtheorem{theorem}{Theorem}[section]

\newtheorem{definition}[theorem]{Definition}

\renewcommand{\theequation}{\thesection.\arabic{equation}}

\newcommand{\ic}{\text{\rm ic}}
\newcommand{\icyes}{\text{\rm ic}_\text{\rm yes}}
\newcommand{\icno}{\text{\rm ic}_\text{\rm no}}

\newcommand{\lb}{\text{$[$}}

\newcommand{\eproof}{\hfill$\square$}

\input xypic

\title{Measuring communication complexity using\\ instance complexity with oracles}

\author{
\authorblockN{Armando Matos}
\authorblockA{LIACC - U.Porto\\
acm@dcc.fc.up.pt}
\and
\authorblockN{Andreia Teixeira}
\authorblockA{LIACC - U.Porto\\
andreiasofia@ncc.up.pt}
\and
\authorblockN{Andr\'{e} Souto}
\authorblockA{LIACC - U.Porto\\
andresouto@dcc.fc.up.pt}
}
\begin{document}

\maketitle


\begin{abstract}
We establish a connection between non-deterministic communication complexity and instance complexity, a measure of information based on algorithmic entropy. Let $\overline{x}$, $\overline{y}$ and $Y_1(\overline{x})$ be  respectively the input known by Alice, the
input known by Bob, and the set of all values of $y$ such that 
$f(\overline{x},y)=1$; a string is a witness of the non-deterministic
communication protocol iff it is a program $p$ that ``corresponds
exactly'' to the instance complexity $\ic^{f,t}(\overline{y}:Y_1(\overline{x}))$.


%
\end{abstract}



\section{Introduction} \label{intro}

In a general scenario of communication complexity there are two parties, Alice and Bob, and the goal is to find the minimal quantity of information, measured in number of bits, that they must exchange in order to compute the value of a given function of their inputs, $f:\,X\times Y\to\{0,1\}$. The instance complexity $\ic^{O,t}(x:A)$, is a rigorous measure of information, based on algorithmic entropy, is the length of the shortest program with access to the oracle $O$ that, in time $t$,
\begin{enumerate}
	\item answers correctly the question ``$x\in A$?'';
	\item does not ``lie'' about the set~$A$ (the program may however answer ``I don't know'' by outputting $\bot$).
\end{enumerate}
Thus, the communication complexity measures the communication costs while instance complexity is related with computational complexity. The objective of this paper is to establish a relationship between these two apparently unrelated measures of complexity.

Let $\overline{x}$ and $\overline{y}$ be the inputs of size $n$ of Alice and Bob respectively. Consider $Y_1(\overline{x})$, the set of all possible inputs $y$ given to Bob such that $f(\overline{x},y)=1$. We prove that, apart from a constant, $\max_{|x|=|y|=n}\{\icyes^{f,t}(\overline{y}:Y_1(\overline{x}))\}$, where $\icyes^{f,t}$ is a ``one-sided'' version of instance complexity is equal to the non-deterministic communication 
complexity $N^1(f)$; as a consequence of this result the maximum value
of $\ic^{f,t}(\overline{y}:Y_1(\overline{x}))$ over all inputs $(\overline{x},\overline{y})$ equals the non-deterministic communication complexity~$N(f)$. The main ingredient for the proof of this result is a protocol in which Alice uses the non-deterministic word $p$ as a program that eventually corresponds to
$\ic^{f,t}(\overline{y}:Y_1(\overline{x}))$. It is important to notice that neither Alice nor Bob alone (i.e., without communication and without the help of the oracle $f$) can compute $\ic^{f,t}(\overline{y}:Y_1(\overline{x}))$; the reason is that Alice only knows $\overline{x}$ and Bob only knows~$\overline{y}$.

We mention two previous works where the communication complexity has been analyzed in a non-standard way: the paper~\cite{BKVV} on individual communication complexity in which Kolmogorov complexity is used as the main analysis tool and~\cite{LR} where ``distinguishers'' are used to obtain bounds on communication complexity.

Our results use a bounded resource version of instance complexity with access to an oracle. Notice that, in the communication complexity scenario, the time of the computations performed by each party is irrelevant. The program $p$ used as a guess must have access to the description of $f$; however, the description of non-uniform functions $f$, which in general is infinite, can not be incorporated into a, necessarily finite, program $p$. Our solution to this problem is based on an oracle which, for each size $n$, gives to $p$ a description of $f$ restricted to inputs $x$ and $y$ of length $n$ (which is of course finite).\footnote{This fact improves the results proved on a previous version of this paper presented at CiE 2007.} We will show that the program $p$ used as a guess must have access to the description of $f$ and so, if $f$ is not uniform, and $p$ does not have access to its description for free, then $p$ would have to built in the description of $f$, which is only possible if its length is unlimited.

The rest of the paper is organized as follows. The next section contains some background and notation on communication complexity and instance complexity. In Section \ref{onesided} we study the one-sided protocols and in Section \ref{twosided} we focus on two sided protocols. These two sections contain the main results of this paper, namely Theorems~\ref{ccyes} and~\ref{ccgen}. Section \ref{kolgo2} contains some comments on the relationship between individual communication complexity and instance complexity.

\section{Preliminaries}\label{prel}
In the rest of this work, $\ene$ denotes the set of natural numbers (including~0). The alphabet that we will be using is $\{0,1\}$. A word is this alphabet is a sequence (possibly empty) of 0's and 1's and will be denoted by $x$, $y$ and $w$, possibly overlined. The length and the $i$-th bit of $x$ are denoted by $|x|$ and $x_i$ respectively.

\subsection{Communication complexity}\label{ccd}

We introduce the basic concepts of communication complexity. For more detailed information see, for example, \cite{KN}. Let $f:\{0,1\}^n\times\{0,1\}^n \to \{0,1\}$ be a boolean function. Alice and Bob want to determine the value $f(\overline{x},\overline{y})$ where $\overline{x}$ is only known by Alice and $\overline{y}$ is only known by Bob. To achieve the goal is imperative that Alice and Bob communicate. Non deterministic protocols $P$ for $f$ involve the usage of a ``guess'' which is given to Alice and Bob. These protocols $P$ satisfy, for $z=0$ or $1$ the following conditions

\vspace*{-3mm}
{\footnotesize\begin{eqnarray}
   [f(\overline{x},\overline{y})=z] & \Rightarrow & \label{c21}
      [\exists w:\, P(w,\overline{x},\overline{y})=z]     \\
\lb f(\overline{x},\overline{y})\not=z] & \Rightarrow & \label{c24}
      [\forall w:\, P(w,\overline{x},\overline{y})\neq z]
\end{eqnarray}}
For~$z\in\{0,1\}$ a ``one-sided'' protocol~$P^z$ has output either~$z$
or~$\bot$ and satisfies 
{\footnotesize\begin{eqnarray}
   [f(\overline{x},\overline{y})=z] & \Rightarrow & \label{c11}
      [\exists w:\, P^z(w,\overline{x},\overline{y})=z]  \\
\lb f(\overline{x},\overline{y})\neq z] & \Rightarrow & \label{c12}
      [\forall w:\, P^z(w,\overline{x},\overline{y})= \bot] 
\end{eqnarray}}
It is easy to build a non-deterministic protocol for~$f$ using the
one-sided protocols~$P^0$ and~$P^1$.

It is important to notice that at the end of any protocol, Alice and Bob must be convinced about the veracity of the value produced, in the sense that ``false guesses'' must be detected and rejected (output~$\bot$). This requirement corresponds to the ``$\forall\cdots$'' predicates above. In other words, this means that Alice and Bob do not trust the oracle. Notice that if both Alice and Bob trusted the oracle the problem would be trivially solved by sending the bit corresponding to the value of the function on their input. 

\begin{definition}[non-deterministic communication complexities]
\label{ndcpx}
Standard and individual (non-deterministic) communication complexities are denoted by~$N$ and~$\cal N$ respectively.
\begin{itemize}
\item [--] {\em Individual communication complexity of protocol $P$ with output set~$\{1,\bot$\}}: \label{iccp1} 
$ {\cal N}_P^1(f,\overline{x},\overline{y}) = \min_w\{|c(w)| : P(w,\overline{x},\overline{y})=1 \}$  where~$w$ is the guess and~$c(w)$ (``conversation'') is the sequence of bits exchanged between Alice and Bob when the guess is~$w$.  Notice that ${\cal N}_P^1(f,\overline{x},\overline{y})$ is only defined if~$f(\overline{x},\overline{y})=1$. Notice also that the behavior of the protocol~$P$ for the other inputs~$(x,y)$ is irrelevant.

\item [--]   {\em Individual communication complexity with output set~$\{1,\bot$\}}: ${\cal N}^1(f,\overline{x},\overline{y}) = \min_P\{{\cal N}_P^1(f,\overline{x},\overline{y})\}$ where the protocols $P$ considered for minimization are one-sided protocols with output set~$\{1,\bot\}$ for the function~$f$.

\item [--] {\em Communication complexity of protocol~$P$ with output set~$\{1,\bot$\}}: $N_P^1(f) =  \max_{\overline{x},\overline{y}}  \{{\cal N}_P^1(f,\overline{x},\overline{y})\}$.

\item [--] {\em Communication complexity of function~$f$ with output set~$\{1,\bot$\}}:
$N^1(f) = \min_{P} \{N_P^1(f)\}$.
\end{itemize}

The complexities~${\cal N}_P^0(f,\overline{x},\overline{y})$, $N_P^0(f)$, and $N^0(f)$, are defined in a similar way.

Define also ${\cal N}_P(f,\overline{x},\overline{y}) = {\cal N}_P^0(f,\overline{x},\overline{y})$ if $f(\overline{x},\overline{y})=0$ and ${\cal N}_P(f,\overline{x},\overline{y}) = {\cal N}_P^1(f,\overline{x},\overline{y})$ if $f(\overline{x},\overline{y})=1$; $N_P(f) = \log(2^{N_P^0(f)}+2^{N_P^1(f)})$; $N(f)=\min_P\{N_P(f)\}$.

A {\em witness} is a guess that causes the protocol to output a value different from~$\bot$.
\eproof
\end{definition}
%

The following result from~\cite{KN} proves that for every function there is a simple optimal non-deterministic protocol.
\begin{theorem}
For every boolean function~$f$ there is an optimal one-sided
non-deterministic protocol~$P$ for~$f$, that is, a protocol~$P$ such
that $N_P^1(f)=N^1(f)$, with the following form where
the witness~$w$, $1\leq w\leq m$, is the index of the
first rectangle~$R_w=A\times B$
containing~$(\overline{x},\overline{y})$ in the first minimum
1-cover:~
\begin{enumerate}
\item ~Alice guesses~$w$ and checks if~$\overline{x}\in A$.
\item ~Alice sends~$w$ to Bob.
\item ~Bob checks if~$\overline{y}\in B$.\eproof
\end{enumerate}

\end{theorem}
Define the sets:
{\footnotesize
\[\begin{array}{lll}
X_0(\overline{y}) = \{x : f(x,\overline{y})=0\}, &
X_1(\overline{y}) = \{x : f(x,\overline{y})=1\},\\
Y_0(\overline{x}) = \{y : f(\overline{x},y)=0\}, &
Y_1(\overline{x}) = \{y : f(\overline{x},y)=1\}.
\end{array}\]}
Notice that Alice knows~$Y_0(\overline{x})$ and~$Y_1(\overline{x})$
while Bob knows~$X_0(\overline{y})$ and~$X_1(\overline{y})$. The
set~$Y_1$ is often mentioned in this paper. 


%
\begin{definition}
  A function is {\em uniform} if it is computed by a fixed (independent of the length of the input) algorithm.
\eproof
\end{definition}
Every function that can be described by an algorithm is uniform; for instance equality and parity are uniform functions. An example of a function which with almost certainty is not uniform is the random function defined as~$f(x,y)=0$ or~$f(x,y)=1$ with probability~1/2. Notice that in the case that $f$ is uniform we can built in the program $p$ a description of $f$ with a small cost (a constant number of bits) in the length of program. On the other hand, if $f$ is not uniform, then the description of $f$ is no longer a constant. To avoid programs of high length for non-uniform functions we allow the program to have oracle access to the description of $f$.
%
\subsection{Instance complexity}
\label{icd}
Instance complexity is a rigorous measure of information of a string relatively to the belonging to a set $A$. It is based on algorithmic entropy, which is, up to a constant term, equal to the expected value of Shannon entropy. We define several forms of instance complexity; for a more complete
presentation see~\cite{OKSW}. It is assumed that programs always
terminate, and output either~0, 1 or~$\bot$ (``don't know''). In the communication complexity the 'cost' is the number of bits exchanged between Alice and Bob who have unlimited computational power. In order to establish a relationship with instance complexity we use a time bounded version of instance complexity where it is assume that the time is sufficiently large.

\begin{definition}
A program~$p$ is {\em consistent} with a set~$A$ if 
$x\in A$ whenever $p(x)=1$ and $x\not\in A$ whenever $p(x)=0$.
\eproof
\end{definition}

\begin{definition}[time bounded instance complexity] \label{ic-def}
Let $t$ be a constructible time bound, $A$ be a set, $x$ an element and $p$ a total program with access to an oracle $O$. Consider the following conditions:
(C1) for all~$y$, $p(y)$ runs in time not exceeding~$t(|y|)$; (C2) for all~$y$, $p(y)$ outputs~0, 1 or~$\bot$, (C3) $p$ is consistent with~$A$ and (C4) $p(x)\neq\bot$.
The {\em $t$-bounded instance complexity with oracle access to $O$} of~$x$ relative to the set~$A$ is 

{\footnotesize
$$
  \ic^{O,t}(x:A) =  \min\left\{|p| : \begin{array}{c}\text{$p$ is a total program with oracle}\\\text{access to $O$ that satisfies (C1), }\\ \text{(C2), (C3) and (C4)}\end{array}\right\}
  $$}

We say that a program $p$ {\em corresponds} to~$\ic^{O,t}(x:A)$ if it satisfies conditions (C1), (C2), (C3) and (C4); if moreover~$|p|=\ic^{O,t}(x:A)$ we say that~$p$ {\em corresponds exactly} to~$\ic^{O,t}(x:A)$.\eproof
\end{definition}

Notice that in the communication complexity the time is not an important issue since Alice and Bob have unlimited power of computation and the communication complexity is measured in number of bits exchanged and not by the time required to transmit the information. The reason why we consider a time bound version of instance complexity is because Alice must have a reference for the time that she can expect for the program, that is given to her as a guess, to stop. This is a technical detail. Notice that if the possible guess $p$ is not a total program then there are data for which the $p$ will not stop and then Alice cannot compute the set of $y$ such that $f(\bar x, y)\not=\bot$, unless she can compute the Halting problem.\\
%


Relaxing the condition ``$p(x)\neq\bot$'' we get two weaker forms of
instance complexity:
\begin{definition} [inside instance complexity] \label{iic}
Let $t$ be a constructible time bound, $A$ be a set, $x$ an element and $p$ a total program with access to an oracle $O$. Consider the following conditions: (C1)~for all~$y$, $p(y)$ runs in time not exceeding~$t(|y|)$, (C2)~for all~$y$, $p(y)$ outputs either~1 or~$\bot$, (C3)~$p$ is consistent with~$A$ and (C4)~$x\in A\Rightarrow p(x)=1$.

The {\em $t$-bounded inside instance complexity with oracle access to $O$} of~$x$ relative to the set~$A$ is 

{\footnotesize$$
\icyes^{O,t}(x:A) = \min\left\{|p| : \begin{array}{c}\text{$p$ is a total program with oracle}\\\text{access to $O$ that satisfies (C1), }\\ \text{(C2), (C3) and (C4)}\end{array}\right\}
$$ }

A program $p$ {\em corresponds} to~$\icyes^{O,t}(x:A)$ if it satisfies
conditions (C1), (C2), (C3) and (C4); if moreover~$|p|=\icyes^{O,t}(x:A)$ we say that~$p$ {\em corresponds exactly} to~$\icyes^{O,t}(x:A)$.
\eproof
\end{definition}
\begin{definition} [outside instance complexity]
Let $t$ be a constructible time bound, $A$ be a set, $x$ an element and $p$ a total program with access to an oracle $O$. Consider the following conditions: (C1)~for all~$y$, $p(y)$ runs in time not exceeding~$t(|y|)$, (C2)~for all~$y$, $p(y)$ outputs either~0 or~$\bot$, (C3)~$p$ is consistent with~$A$ and (C4)~$x\not\in A\Rightarrow p(x)=0$.

The {\em $t$-bounded outside instance complexity with oracle access to $O$} of~$x$ relative to the set~$A$ is 
{\footnotesize
$$\icno^{O,t}(x:A) = \min\left\{|p| : \begin{array}{c}\text{$p$ is a total program with oracle}\\\text{access to $O$ that satisfies (C1), }\\ \text{(C2), (C3) and (C4)}\end{array}\right\}
$$}

A program~$p$ {\em corresponds} to~$\icno^{O,t}(x:A)$ if it satisfies
conditions (C1), (C2), (C3) and (C4); if moreover~$|p|=\icno^{O,t}(x:A)$ we say that~$p$ {\em corresponds exactly} to~$\icno^{O,t}(x:A)$.
\eproof
\end{definition}

Notice that if~$x\not\in A$ then $\icyes^{O,t}(x:A)$ is a constant (independent of~$x$) for a time bound $t$ (namely a constant), because the program~$p(x)\equiv\bot$ has fixed length and is consistent with every set; similarly if~$x\in A$ then $\icno^{O,t}(x:A)$ is a constant. Notice also that for every element~$x$ and set~$A$ we have $\icyes^{O,t}(x:A) \leq \ic^{O,t}(x:A)$ and $\icno^{O,t}(x:A) \leq \ic^{O,t}(x:A)$.

On the other hand, from a program~$p$
corresponding to~$\icyes^{O,t_1}(x:A)$ and a program~$p'$ corresponding
to~$\icno^{O,t_2}(x:A)$ we can define a 
program~$r$ as follows:
$r(x)=1$ if $p(x)=1$,
$r(x)=0$ if $p'(x)=0$ and
$r(x)=\bot$ otherwise,
concluding that
{\footnotesize$$\begin{array}{ll}
   \ic^{O,f(t_1,t_2)}(x:A)& \leq 
   \icyes^{O,t_1}(x:A) + \icno^{O,t_2}(x:A) +\\ & 
   O(\log(\min\{\icyes^{O,t_1}(x:A),\icno^{O,t_2}(x:A)\}))
\end{array}$$}
\vspace*{-2mm}

where the function~$f$ represents the time overhead
needed for the simulation of~$p(x)$ for~$t_1$ steps
followed by simulation of~$p'(x)$ for~$t_2$ steps; the
logarithmic term comes from the need to delimit~$p$ from~$p'$ in the
concatenation~$pp'$.

\section{One-sided protocols} \label{onesided}
To give an idea of the relationship between instance complexity and communication complexity we first analyze, in sub-section~\ref{neq}, the special case of function inequality defined by $NEQ(\overline{x},\overline{y})=1$ if and only if $\overline{x}\neq\overline{y}$. We show how to use programs corresponding to instance complexity as guesses of (optimal) non-deterministic protocols. This usage is later generalized to any function in sub-section~\ref{general}.

\subsection{Inequality: an optimal ``$\icyes$-protocol''}
\label{neq}
Consider the predicate NEQ and suppose that~$\overline{x}\neq\overline{y}$;
then for some~$i$, $1\leq i\leq n$, we have $\overline{x}_i\neq\overline{y}_i$. A possible program~$p_i$
corresponding to $\icyes^{NEQ,t}(\overline{y}:Y_1(\overline{x}))$ is 
$p_i(y) = 1$ if $y_i\neq\overline{x}_i$,
$p_i(y) = \bot$ if $y_i=\overline{x}_i$.
If the reader computes the set $Y'_1=\{y:p_i(y)=1\}$ it is easy to see that~$Y'_1\subset Y_1(\overline{x})$. So, if $p(\overline{y})=1$ and if $|p|$ is minimum, this program corresponds exactly to $\icyes^{NEQ,t}(\overline{y}:Y_1(\overline{x}))$ for some function $t$.

Consider now the following protocol~$P$ for $NEQ$ where~$t$ is a time
bound sufficiently large (see more details in
sub-section~\ref{general}). Alice receives a word~$p$ as a guess; $p$ may eventually be the program~$p_i$ above. Then she runs~$p(y)$ for
every~$y\in Y$ until the program halts or until~$t(|y|)$ steps have
elapsed. If~$p(y)$ does not halt in time~$t(|y|)$, the word~$p$ is not
a valid witness and the protocol halts. Otherwise Alice defines the
set $Y'_1=\{y:\,p(y)=1\}$. If $Y'_1\subseteq Y_1(\overline{x})$, i.e., if~$p$ is consistent with~$Y_1(\overline{x})$, she sends~$p$ to Bob, otherwise outputs~$\bot$ and halts. Bob tests if~$p(\overline{y})=1$; if yes, outputs~1, otherwise outputs~$\bot$.

\noindent\underline{Correctness conditions:}

\begin{enumerate}
\item If~$\overline{x}\neq \overline{y}$, there is a witness~$p$ that corresponds to~$\icyes^{NEQ,t}(\overline{y},Y_1(\overline{x}))$.

We have~$\overline{x}_i\neq\overline{y}_i$ for some~$i$, $0\leq i\leq n$. Then, if~$p$ happens to be the program~$p_i$ above, the protocol~$P$ outputs~1 so~$p$ corresponds to~$\icyes^{NEQ,t}(\overline{y},Y_1(\overline{x}))$, that is, we must have $Y'_1$ consistent with~$Y_1(\overline{x})$ (verified by Alice) and~$p(\overline{y})=1$ (verified by Bob).

\item If a guess is wrong, the output is~$\bot$.

If the guess is wrong, then either some~$p(y)$ does not run in time $t$ or~$p$ is not consistent with~$Y_1(\overline{x})$ or~$p(\overline{y})=\bot$;
All these cases are possible to detect by Alice and Bob.

\item  If~$\overline{x}= \overline{y}$, no guess~$p$ can cause output~1.

This follows directly from the definition of the protocol. 
\end{enumerate}

\noindent\underline{Complexity:}

The length of~$p_i$ need not to exceed $\log n+O(1)$ and $\max_{0\leq i\leq n}\{|p_i|\}$ is~$\log n +O(1)$. Thus the complexity of the protocol~$P$ is~$\log(n)+O(1)$. But the non-deterministic communication complexity of $NEQ$ is also $\log n+O(1)$ (see~\cite{KN}), thus the protocol is optimal.

\subsection{``$\icyes^{f,t}$-protocols'' are optimal}
\label{general}

In this section we prove the main theorem of this paper by showing how to use a program corresponding to $t$-bounded inside instance complexity as a guess in a nondeterministic protocol. In the general case, the function~$f$, which is known by Alice and Bob, is arbitrarily complex; therefore the description of~$f$ can not be included into an ``instance complexity program''~$p$ unless~$\lim_{n\to\infty}|p|=\infty$. But the scenario is different if we give the program $p$ free access to the description of the function $f$.

\begin{theorem}[$\icyes^{f,t}$-protocols are optimal]\label{ccyes}
  Let~$f$ be an arbitrary function. There is a computable
  function~$t(n)$ such that
{\footnotesize  \begin{eqnarray}
    \label{ccyest}
    N^1(f) = \max_{|x|=|y|=n}\{\icyes^{f,t}(y:Y_1(x))\}  + O(1)
  \end{eqnarray}}
\end{theorem}

\begin{proof}

Let~$p$ be the non-deterministic word given to Alice by the third entity;
the protocol~$P$ is described in Figure~\ref{side1}. Notice that the protocol specifies that Alice should interpret~$p$ as a program and execute~$p$ for all $y$ for $t(|y|)$ steps.

\vspace*{-2mm}
\begin{figure}[h]
\begin{center}

{\footnotesize
\begin{tabular}{|p{0.1 cm}p{8 cm}|}  \hline
  {\bf Alice:} &                                              \\
  &{\tt   Receive program~$p(y)$ (as a possible witness)}     \\
  &{\tt   Test if, for every~$y\in Y$, $p$ halts and produces $\bot$ or $1$ in time $t(n)$}\\
  &{\tt   If not, output~$\bot$ and halt}                     \\ 
  &{\tt   Compute the set $B=\{y:p(y)=1\}$}                   \\
  &{\tt   Using the oracle access to the description of $f$}\\
  &{\tt \hspace{2mm} find the set of smallest 1-covers}                  \\
  &{\tt   Select the first (in lexicographic order) such }\\ 
  &{\tt \hspace{2mm} cover~$\langle R_1, R_2,\ldots R_m\rangle$}         \\
  &{\tt   Select a rectangle~$R_i = A\times B$ from that cover}\\
  &{\tt   \hspace{2mm}
          where~$B\subseteq Y$ is the set computed above}     \\
  &{\tt   \hspace{2mm}
          As the cover is minimum, there can be at most one such rectangle. If there is none, output~$\bot$ and halt}\\

  &{\tt   Test if~$\overline{x}\in A$}                        \\
  &{\tt   If not, output~$\bot$ and halt}                      \\
  &{\tt   Send~$p$ to Bob}                                    \\ \hline
  {\bf Bob:} &                                                \\
  &{\tt   Verify if~$p(\overline{y})=1$}                      \\
  &{\tt   If yes, output~1 and halt}                          \\
  &{\tt   Output~$\bot$ and halt}                             \\ \hline
\end{tabular}
}
\vspace*{-5mm}
\end{center}
\caption{\footnotesize
  A family of one-sided non-deterministic protocols~$P$.
  The guess is based on a program~$p$ that corresponds to~$\icyes^{f,t}(\overline{y}:Y_1(\overline{x}))$. 
}  
\label{side1}
\end{figure}
\vspace*{-0mm}
The program~$p$, being an arbitrary guess, may behave in many different ways; in particular, if~$f(\overline{x},\overline{y})=1$, the behavior can be described as follows:


\begin{figure}[h]
\begin{center}
{\footnotesize
\begin{tabular}{|p{8.2 cm}|}                                      \hline
  {\bf Program~$p$, input~$y$:}                               \\ \hline
  {\tt   From the description of $f$ (which is given by the oracle) and $i$:}                \\
  {\tt   \hspace{1mm} Find the set~$S_1$ of smallest 1-covers}\\
  {\tt   \hspace{1mm} 
         Select the first (in lexicographic order) cover}\\
   {\tt \hspace{3mm}  $\langle R_1, R_2,\ldots R_m\rangle\in S_1$}     \\
  {\tt   \hspace{1mm}  Select rectangle~$R_i = A\times B$ 
                    in that cover}                            \\
  {\tt   With input~$y$, output}                              \\
  {\tt   \hspace{2mm} $p(y)=1$    if ~$y\in B$}               \\
  {\tt   \hspace{2mm} $p(y)=\bot$ otherwise}                  \\ \hline
\end{tabular}
\vspace*{-3mm}
}\end{center}

\caption{\footnotesize A possible behavior of the program~$p$ which
  may cause the protocol~$P$ (see Figure~\ref{side1}) to
  output~1.  A string~$p$ with this behavior can be specified in
  length~$\log m$.  The existence of this program, which has
  length~$\log m$ where~$m$ is the size of the minimum covers,
  justifies the step between equation~(\ref{logm}) and
  inequality~(\ref{max}).
}
\label{pprog}
\end{figure}
\vspace*{-0mm}

  If~$i$ is chosen so that~$(\overline{x},\overline{y})\in R_i$
  (if~$f(\overline{x},\overline{y})=1$ there is at least one such~$i$,
  otherwise there is none) then~$p$ is consistent
  with~$Y_1(\overline{x})$ and~$p(\overline{y})=1$. Then $|p|\geq
  \icyes^{f,t}(\overline{y}:Y_1(\overline{x}))$. Moreover, if~$p$ is not ``correct'', that fact can be detected by Alice or by Bob; thus, conditions~(\ref{c11}) and~(\ref{c12}) (see page~\pageref{c11}) are verified.

How much time~$t(n)$ must Alice run~$p(y)$ (for each~$y$) so that, there is at least a witness for every pair~$(\overline{x},\overline{y})$ with
$f(\overline{x},\overline{y})=1$? It is possible to obtain an upper bound~$t(n)$ in a constructive way by detailing and analyzing the algorithm that the witness~$p$ should implement, see Figure~\ref{pprog}. In fact, $t(n)$ is a computable function that Alice can determine.\footnote{Notice that the time required must be, at least, exponential since the determination of the minimal cover can be determine in exponential time.}
   Suppose now that~$f(\overline{x},\overline{y})=1$. If the protocol
   accepts~$(\overline{x},\overline{y})$ with guess~$p$,
   we have~$|p|\leq \log m + O(1)$ and
   $\max_{|\overline{x}|=|\overline{y}|=n}\{|p|\}\leq \log m + O(1)$. Thus
{\footnotesize  \begin{eqnarray}
       N^1(f)  &=&   \log C^1(f) + O(1) \\ 
               &=&   \log m  + O(1)     \label{logm}     \\
            &\geq&   \max_{|\overline{x}|=|\overline{y}|=n}\{|p|\}
                                + O(1)  \label{max}      \\
            &\geq&   \max_{|\overline{x}|=|\overline{y}|=n}
                     \{\icyes^{f,t}(\overline{y}:Y_1(\overline{x}))\}
                                + O(1) \label{last}
  \end{eqnarray}}
  On the other hand, {\em there exists} a non-deterministic protocol
  with complexity  $\max_{|\overline{x}|=|\overline{y}|=n}\{
    \icyes^{f,t}(\overline{y}:Y_1(\overline{x}))\}+O(1)$; 
  this is the protocol of Figure~\ref{wit-ic-yes}.
Notice that program~$p$ can be {\em any} total program running in time $t$ which is consistent with~$Y_1(\overline{x})$ and such that~$p(\overline{y})=1$ (and, if $f(\overline{x},\overline{y})=1$, there is at least one such program, as we have seen above); thus it can be the shortest such program, $|p|=\icyes^{f,t}(\overline{y},Y_1(\overline{x}))$. Taking the maximum over all~$\overline{x}$ and $\overline{y}$ with
$|\overline{x}|=|\overline{y}|=n$ (see Definition~\ref{ndcpx}) we get
$
   N^1(f) \leq \max_{|\overline{x}|=|\overline{y}|=n}\{
     \icyes^{f,t}(\overline{y}:Y_1(\overline{x}))\} + O(1)
$ 
because~$N^1(f)$ is the smallest complexity among all the protocols
for~$f$. Combining this result with inequation~(\ref{last}) we 
get~$N^1(f) =  \max_{|\overline{x}|=|\overline{y}|=n}\{
    \icyes^{f,t}(\overline{y}:Y_1(\overline{x}))\}+O(1)$.
\end{proof}

\vspace*{-5mm}
\begin{figure}[bht]
\begin{center}
{\footnotesize
\begin{tabular}{|p{0.1 cm}p{8 cm}|}  \hline
  {\bf Alice:}  &  \\
    &  {\tt Receive program~$p(y)$ (as a possible witness)}     \\
    &  {\tt Test if, for every~$y$, 
                           $p(y)$ halts in time $t$}\\
    &  {\tt If not, output~$\bot$ and halt}  \\
    &  {\tt Test if~$\{y:p(y)=1\}\subseteq Y_1(\overline{x})$}\\
    &  {\tt ($p$ is consistent with $Y_1(\overline{x})$)}    \\
    &  {\tt If not, output~$\bot$ and halt}  \\
    &  {\tt Send~$p$ to Bob}         \\
  {\bf Bob:}  &  \\
    &  {\tt Compute~$r=p(\overline{y})$ and test if~$r=1$}  \\
    &  {\tt If not, output~$\bot$ and halt}  \\
    &  {\tt Output~1}                        \\
  {\bf Alice:}  &  \\
    &  {\tt Output 1}     \\ \hline
\end{tabular}
\vspace*{-5mm}
}\end{center}
\caption{\footnotesize
  A family of one-sided non-deterministic
  protocols~$P'$. The guess may be any program~$p$ that
  corresponds to~$\icyes^{f,t}(\overline{y}:Y_1(\overline{x}))$, that is
  $p$ must satisfy only~$\{y:p(y)=1\}\subseteq Y_1(\overline{x})$
  and~$p(\overline{y})=1$.
} 
\label{wit-ic-yes}
\end{figure}
\vspace*{-5mm}

\subsection*{A note on the uniformity condition}\label{unif}
At first it may not be obvious why the validity of equality~(\ref{ccyest}) of Theorem~\ref{ccyes} depends on the fact that $p$ has access to the description of $f$. In what follows we show that if $f$ is not uniform then~(\ref{ccyest}) may be false. Notice that if we do not allow access to an oracle the result is valid for uniform functions since the description of $f$ in this cases requires a constant number of bits and hence can be built in the program that is used as a guess with a cost of a constant in the number of bits. The idea to prove that the result is false without oracle access is to use the Kolmogorov complexity as a tool.
Denote by~$C(x)$ the (plain) Kolmogorov complexity of~$x$ which
is defined as $C(x)=\min\{|p|:U(p)=x\}$ where~$U$ is some fixed
universal Turing machine, see~\cite{LV}. 

Consider a monochromatic cover of a non uniform function such that
(i)~the number~$m$ of rectangles in the cover is very small and
(ii)~the horizontal side~$B$ of the first rectangle in the cover has a
Kolmogorov random length, $C(|B|)\approx n$. The length~$B$ can be
obtained from~$p$, thus $C(|B|)\leq C(p)+O(1)$ which
implies~$C(p)\geq n+O(1)>\!\!> \log m$; thus the step
(\ref{logm})$\,\to\,$(\ref{max}) in the proof is not valid.



\section{Two-sided protocols}   \label{twosided}
Now we consider the two-sided protocols for non-deterministic
communication complexity. Similarly to the result of the previous section we show that there are optimum protocols whose guesses correspond exactly 
to~$\ic^{f,t}(\overline{y}:Y_1(\overline{x}))$.

\begin{theorem} \label{ccgen}
  Let~$f$ be any function. There exists a computable function $t$ such that
{\footnotesize  $$
    N(f) = \max_{|\overline{x}|=|\overline{y}|=n}\{\ic^{f,t}(y:Y_1(x))\}  + O(1)
  $$}
\end{theorem}

The proof of this Theorem is similar to the proof of Theorem~\ref{ccyes}; we make only a few observations. The reader should compare Figures~\ref{side1} and~\ref{pprog} with Figures~\ref{side2} and~\ref{pside2} respectively. The main difference in the proof is that we have now to consider a minimum cover of 0-rectangles and a minimum cover of 1-rectangles. Denote by~$m=C^0(f)$ and~$m'=C^1(f)$ the size of those covers; the witness (program)~$p$ has a description with length~$\log(m+m')+O(1)$.  It is not difficult to verify the correctness of conditions~(\ref{c21}) to~(\ref{c24}), see page~\pageref{c24}.

\begin{figure}[h]
\begin{center}
{\footnotesize
\begin{tabular}{|p{0.1 cm}p{8 cm}|}  \hline
  {\bf Alice:} &                                              \\
  &{\tt   Receive program~$p(y)$ (as a possible witness)}     \\
  &{\tt   Test if, for every~$y\in Y$, $p(y)$ halts in $t(n)$ steps with output~0, 1 or~$\bot$}          \\
  &{\tt   Compute the set $B=\{y:p(y)\neq\bot\}$}             \\
  &{\tt   Test if~$B$ is monochromatic and not empty}         \\
  &{\tt   Using the description of $f$, find the set~$S_0$ of}\\
  &{\tt  \hspace{2mm} smallest 0-covers and the set~$S_1$ of smallest}\\
  &{\tt  \hspace{2mm}  1-covers}            \\
  &{\tt   Select the first (in lexicographic order) sequence $s=\langle R_1, \ldots R_m,R_{m+1},\ldots R_{m+m'}\rangle$}\\
  &{\tt   \hspace*{2mm}
            where $\langle R_1, \ldots R_m\rangle\in S_0$
            and   $\langle R_{m+1}, \ldots R_{m+m'}\rangle\in S_1$
          } \\
  &{\tt   Select a rectangle~$R_i = A\times B$ from~$s$}      \\
  &{\tt   \hspace*{2mm} \footnotesize{\underline{Comment.} There is at most one such rectangle}
                                                             }\\
  &{\tt   Test if~$\overline{x}\in A$}                       \\
  &{\tt   Send~$p$ to Bob}                                    \\ \hline
  {\bf Bob:} &                                                \\
  &{\tt   Compute $r=p(\overline{y})$}                        \\
  &{\tt   Output $r$}                                         \\ \hline
\end{tabular}
\vspace*{-5mm}
}\end{center}
\caption{\footnotesize
  A family of two-sided non-deterministic protocols~$P$. The guess is based on a program~$p$ that corresponds to~$\ic^{f,t}(\overline{y}:Y_1(\overline{x}))$. Compare with Figure~\ref{side1}. For simplicity we assume that whenever a test fails, the protocol outputs~$\bot$ and halts.}
\label{side2}
\end{figure}
\vspace*{-5mm}
\begin{figure}[h]
\begin{center}
{\footnotesize
\begin{tabular}{|p{8.7 cm}|}                                            \hline
  {\bf Program~$p$, input~$y$:}                              \\ \hline
  {\tt   From the description of $f$ given by the oracle and~$i$:}    \\
  {\tt   \hspace{2mm} 
         Find the set~$S_0$ of smallest 0-covers and}        \\
    {\tt   \hspace{6mm} the set~$S_1$ of smallest 1-covers}  \\
  {\tt   \hspace{2mm} 
         Select the first (in lexicographic order)} \\  
  {\tt   \hspace*{4mm}
           sequence $s=\langle R_1, \ldots R_m,R_{m+1},\ldots R_{m+m'}\rangle$}\\
  {\tt   \hspace*{4mm}
            where $\langle R_1, \ldots R_m\rangle\in S_0$
            and   $\langle R_{m+1}, \ldots R_{m+m'}\rangle\in S_1$
          } \\
  {\tt   \hspace{2mm} 
         Select the $i$th rectangle~$R_i = A\times B$ from~$s$}\\
  {\tt   With input~$y$, output:}                              \\
  {\tt   \hspace{2mm} $p(y)=z$    if ~$y\in B$
         and rectangle $A\times B$ has color~$z\in\{0,1\}$}    \\
  {\tt   \hspace{2mm} $p(y)=\bot$ otherwise}                  \\ \hline
\end{tabular}
\vspace*{-5mm}}
\end{center}
\caption{\footnotesize A possible behavior of the program~$p$ which
  may cause the protocol~$P$ of Figure~\ref{side2} to
  output a value different from~$\bot$. A string~$p$ with this
  behavior can be specified in length~$\log (m+m')$.
}
\label{pside2}
\end{figure}
\vspace*{-3mm}
\section{About individual communication complexity}
\label{kolgo2}
The one sided individual communication complexity satisfies
{\footnotesize
$$
    {\cal N}^1(f,\overline{x},\overline{y}) \geq 
    \icyes^{f,t}(\overline{y}:Y_1(\overline{x}))  + O(1)
$$ 
}
for some constructible time $t$. The complexity
${\cal N}^1(f,\overline{x},\overline{y})$ is obtained from a
minimization over all protocols which must of course ``work
correctly'' for every pair~$(x,y)$ and not only
for~$(\overline{x},\overline{y})$ while no such restriction exists in
the definition of  instance complexity. The individual communication
complexity may in a few rare cases (if~$i$ has a very short
description), be much smaller than~$\log m$.

Finally we present a result relating the individual
non-deterministic communication complexity with the instance
complexity.
\begin{theorem}(Individual upper bound) \label{iccub}
For every function~$f$ and values~$x$ and~$y$ the individual non-deterministic communication complexity~${\cal N}(f,{x},{y})$ satisfies for some constructible time $t$ 

{\footnotesize
$$
{\cal N}(f,{x},{y}) = \ic^{f,t}({y}:Y_1({x})) + O(1) \leq N(f) + O(1) 
$$}
\end{theorem}

%
%
%

\section*{Acknowledgments}
{\footnotesize
We thank Sophie Laplante for helpful discussions and suggestions. The authors are partially supported  by KCrypt (POSC/EIA/60819/2004) and funds granted to LIACC  through   the  Programa  de Financiamento Plurianual, Fundaç\~{a}o   para  a  Ciência  e Tecnologia  and Programa POSI.
}


\begin{thebibliography}{9}

\bibitem{AB}
S. Arora, B. Barak,
  {\em Computational Complexity: A Modern Approach},
  Princeton University, 2006,

\bibitem{BKVV}  H. Buhrman, H. Klauck, N. Vereshchagin
  and P. Vitány, 
  {\em Individual communication complexity},
  Proc. of STACS 2004.

\bibitem{KNR}
  I. Kremer, N. Nisan, D. Ron,
  {\em On randomized one-round communication complexity},
   Proc. of STOC, pp~596-605, 1995,

\bibitem{KN}
  Eyal Kushilevitz, Noam Nisan,
  {\em Communication Complexity},
  Cambridge University Press, New York, 
  Springer-Verlag, 1996.

\bibitem{LR}
  S. Laplante, J. Rogers,
  {\em Indistinguishability},
  TR-96-26,
  1996,

\bibitem{LV}
  M. Li e P. Vitányi,
  {\em An Introduction to Kolmogorov Complexity and its Applications},
  Springer, second edition, 1997.

\bibitem{OKSW}
  P. Orponen, K. Ko, U. Sch{\"{o}}ning, O. Watanabe,
  {\em Instance Complexity},
  Journal of the ACM, 41:1, pp~96-121,
  1994,

\bibitem{yao79}
  A. Yao,
  {\em Some complexity questions related to distributive computing},
  Proceedings of Symposium on Theory of Computing, pp~209-213, 1979.

\bibitem{yao81}
  A. Yao,
  {\em The entropic limitations on VLSI computations},
  Proceedings of Symposium on Theory of Computing,
  pp~308-311, 1981.

\end{thebibliography}
\end{document}